\documentclass[12pt, showpacs, showkeys, nofootinbib]{article}

\usepackage{amsmath, amscd, amssymb, mathrsfs, tikz-cd, natbib, verbatim, amsthm}
\newcommand{\pa} [2] {\frac{\partial #1}{\partial #2}} 

\newtheorem{theorem}{Theorem}[section]

\theoremstyle{proposition}
\newtheorem{proposition}{Proposition}[section]

\theoremstyle{definition}
\newtheorem{definition}{Definition}[section]

\begin{document}

\title{A global version of G\"unther's polysymplectic formalism using vertical projections}
\author{Tom McClain\footnote{Email: mcclaint@wlu.edu}
 \\
Department of Physics and Engineering \\ Washington and Lee University, Lexington Virginia 24450}

\maketitle

\begin{abstract}
I construct a global version of the local polysymplectic approach to covariant Hamiltonian field theory pioneered by C. G\"unther. Beginning with the geometric framework of the theory, I specialize to vertical vector fields to construct the (poly)symplectic structures, derive Hamilton's field equations, and construct a generalized Poisson structure. I then examine a few key examples to determine the nature of the necessary vertical projections and find that the theory seems to provide the geometric analog of the canonical transformation approach to covariant Hamiltonian field theory advanced by Struckmeier and Redelbach. I conclude with a few remarks about possible applications of this framework to the geometric quantization of classical field theories.
\end{abstract}

\section{Introduction}

The study of covariant Hamiltonian field theory goes back at least to the pioneering works of Dedonder \cite{de1930theorie} and Weyl \cite{weyl1935geodesic} in the 1930s. Geometric approaches to covariant Hamiltonian field theory really began to be pursued by Dedeker \cite{dedecker1977generalization} and Goldschmidt and Sternberg \cite{goldschmidt1973hamilton} in the 1970s. This work was then taken up by G\"unther in the 1980s, and given a set of clear axiomatic foundations \cite{gunther1987polysymplectic}. Though quite rigorous, G\"unther's original paper only dealt with the case in which the fields are sections of a trivial vector bundle (that is, the local version of the theory). One transition to a global theory for G\"unther's essential method was later provided by Carinena et al.\ in \cite{carinena1991multisymplectic}, and the method was more recently revived and re-formulated in more modern mathematical language by Munteanu et al.\ in \cite{munteanu2004gunther}. In addition to these more-or-less direct extensions, G\"unther's work has been quite influential in the move toward a differential geometric foundation for covariant Hamiltonian field theory generally, as noted in many papers detailing different perspectives (see \cite{sardanashvily1995generalized}, \cite{gotay1998momentum}, and \cite{kanatchikov1998canonical}, to name just a few key examples).

Though G\"unther's work was a major advance in the field, it was indeed only local in character, leaving an important gap that needed to be filled to provide a fully general treatment. In the extension of Carinena et al., the authors deviate substantially from the main path of G\"unther's paper, deeming the approach that more closely follows G\"unther's geometric foundation -- the vector-valued polysymplectic structure -- to be ``artificial and unsatisfactory" \cite{carinena1991multisymplectic} \footnote{A few comments made by those authors just before and after the quoted text shows that they had understood several of the main points of this work. However, they seem never to have presented the details of that work.}. Later expositions and continuations of G\"unther's work like that of \cite{munteanu2004gunther} also forgo its original geometric foundation in favor of other, more modern approaches like k-symplectic structures. Though \cite{blacker2019polysymplectic} provides a much needed update to the original treatment at a greater level of mathematical sophistication, to this author's knowledge nowhere is G\"unther's polysymplectic structure in its original form used to produce the full range of phenomena of covariant Hamiltonian field theory in the general, global case. The first goal of this paper is fill in this missing link.

In doing so, we will find the main critique of Carinena et al.\ in \cite{carinena1991multisymplectic} to be correct: the path from the starting fiber bundle through the polysymplectic structure to Hamilton's field equations, though feasible and well-defined, does indeed require extra structure beyond the original polysymplectic or standard multisymplectic approaches that make it seem artificial. Against this (somewhat subjective) demerit, we must weigh the payoff of this approach: a directly and naturally defined generalized Poisson structure that acts on generic functions on the phase space of the theory. Some of the common difficulties that surround the construction of generalized Poisson structures in covariant Hamiltonian field theory (see, for example, \cite{forger2003poisson}) are bypassed in this approach. Indeed, the outcome of the polysymplectic structure path seems to be the almost exact duplication of the natural results of the canonical transformation approach to covariant Hamiltonian field theory presented by Struckmeier and Redelbach in \cite{struckmeier2008covariant}. A second perspective on this paper, then, is that it presents a natural geometric setting for \cite{struckmeier2008covariant} that directly reproduces one of the main results of that paper. 

One of the longstanding hopes for covariant Hamiltonian field theory is that it will provide a path for the geometric quantization of classical fields  \cite{axelrod1991geometric}, \cite{kanatchikov1998quantization}, \cite{forger2003poisson} (but note also \cite{gotay1980functorial}) and/or especially challenging classical particle systems \cite{guillemin1982geometric},\cite{giachetta2002geometric}. From this point of view a third, more aspirational perspective on the paper is that it presents another geometric framework for covariant Hamiltonian field theory, one which is at least superficially rather different from many of the other modern approaches and which may therefore prove itself to be more applicable to the as-yet-unsolved problem of the geometric quantization of classical fields. I will make some preliminary remarks on this possibility in the conclusion. A more detailed exposition is planned in a forthcoming paper. 

Given these three possible motivations for readers, I have tried to make my treatment as accessible as possible to as wide an audience as possible. I have avoided specialized or esoteric language and operations unless absolutely necessary, and I have explained such non-standard operations or language as I find it essential to employ. Though I give coordinate independent definitions of all structures and operations, I also provide local coordinate descriptions of all results, and the paper can be followed by looking only at those local coordinate results (though in this case much of the motivation will be lost). It is my hope that a reader with a solid grasp of the fundamentals of differential geometry and elementary variational calculus will be able to follow the entire paper. \\

The main new result is the extension of G\"unther's polysymplectic formalism to non-flat space-times through the use of vertical projetion operators (Ehresmann connections). Eqs.\ \eqref{hamiltonsfieldequationsgeometric} and the more physically relevant \eqref{alternativehamiltons} are the most important results, but both are built on the more fundamental results of defining the tautological tensor in terms of a vertical projection operator and the derivation of the resulting polysymplectic structure; see definition \ref{tautologicaltensordefinition} and theorem \ref{symplectictensor} for more details. Of almost equal significance -- especially for quantization -- is the definition of a generalized Poisson structure based on this approach; see especially theorem \ref{poissonstructure} and the preceeding constructions. From the physical perspective, the application of this formalism to scalar and Maxwell fields are the most signficant results; please see sections \ref{scalarfield} and \ref{maxwellfield} for more details. \\

In a field with so many different conventions and notations, a few remarks are in order about which ones I will adopt throughout the paper. Given a fiber bundle $\epsilon : E \to M$, I will denote the jet bundle simply as $JE$ instead of $J^1 E$ (I will not use higher order jet bundles in this paper), and the vector bundle upon which $JE$ is modeled will be called $\vec JE$. Given a map $f: M \to N$, I will denote the tangent map $Tf : TM \to TN$ as $Tf$ rather than $f_*$. I will use the notation $\Gamma(M,E)$ to denote sections of a fiber bundle $E$ with projection map $\epsilon : E \to M$, and will often use notation like $E_x$ to denote the part of one geometric structure (in this case the fiber of $E$) that lies over a specific point of another (in this case the point $x \in M$). I make use of fibered local coordinate systems throughout the paper, but nothing of the geometry is changed by using non-fibered coordinate systems (only the local coordinate descriptions). The Einstein summation convention is employed throughout.

\section{A Tensor Algebra Approach to G\"unther's Formulation} \label{gunther's}
Before beginning to embellish G\"unther's original, flat-space approach, it makes sense to first restate it in the language of tensorial structures similar to those I will eventually use in the non-flat case. This will serve both as a reminder of the main line of reasoning in G\"unther's paper and its results, as well as an introduction to the tensorial approach I will use. None of the material in this section is at all new; it is only a restatement of some of the main symplectic structures of \cite{gunther1987polysymplectic} in slightly different terms. 

As with the geometric approach to Hamiltonian particle theory, in G\"unther's approach the parameter space of the theory is suppressed at the beginning. The approach begins directly with the space $Q$ in which the physical fields under consideration take their values. This space is taken to be a differentiable manifold with local coordinates $\{ q^I \}$. 

The next step is to form the phase space $P$. In geometric Hamiltonian particle theory, this would simply be the cotangent bundle $T^*Q$. However, in field theory the fact that the parameter space -- suppressed though it may be at this stage -- is no longer one dimensional means that care must be taken to produce enough momenta to encompass the covariant dynamics of the theory. To do this, the space is instead taken to be the bundle $P$ over $Q$ with standard fiber
$$ P_q = T^*_q Q \otimes \mathbb{R}^n $$
This bundle comes with a standard projection map $\pi : P \to Q$ and coordinates $\{ q^I, p^i_I \}$, where the capital letter index $I$ runs over field degrees of the freedom and the lower case index $i$ runs over parameter space dimensions. 

Given an element $v \in TP$, the tautological vector-valued one-form $\theta$ is defined point-wise by:
$$ \theta(v)_p := p \circ T_p \pi (v) $$
where $T \pi : TP \to TQ$ is the differential of the projection map $\pi$. In fibered coordinates $\{ q^I, p^i_I, v^I, v^i_I \}$ on $TP$, this map is given by
$$ \theta_p = p^i_I \ dq^I \otimes e_i$$
where the $\{ e_i \}$ are basis vectors for $\mathbb{R}^n$. The polysymplectic form $\omega$ is then given by
$$ \omega := - d \theta $$
The fact that this approach uses $\mathbb{R}^n$ -- a single, fixed vector space -- as the space in which the vector parts of $\theta$ take their values means that the exterior derivative here gives no trouble. See \cite{gunther1987polysymplectic} for details. In fibered coordinates on $TP$, this tensor is given by
$$ \omega = dq^I \wedge dp^i_I \otimes e_i $$ 

Given a Hamiltonian $H : P \to \mathbb{R}$ that encodes the dynamics of the theory, a map $\gamma : \mathbb{R}^n \to P$ represents a physically realizable field configuration if and only if it satisfies the condition
$$ \omega(\iota(T \gamma), -) = dH $$
over every point $p \in \text{Im} \gamma$. Here $T_r \gamma : T_r \mathbb{R}^n \to T_{\gamma(r)} P$ is the differential of the map $\gamma$, and the natural identification $\iota_r : T^*_r \mathbb{R}^n \to (\mathbb{R}^n)^*$ of each cotangent space $T^*_r \mathbb{R}^n$ with the vector space $(\mathbb{R}^n)^*$ is used when contracting the differential of $\gamma$ with the polysymplectic form. In the fibered coordinates used above, this condition amounts to the two equations
$$ \pa{q^I}{x^i} = \pa{H}{p^i_I} \quad , \quad \pa{p^i_I}{x^i} = - \pa{H}{q^I} $$
which are Hamilton's field equations. 

Formulated in this manner, it is not hard to see that one way to extend G\"unther's approach to the case in which the parameter space $\Lambda$ is not simply $\mathbb{R}^n$ is to switch to a setup in which the map $\iota : T^*_\lambda \Lambda \to (\mathbb{R}^n)^*$ is some other isomorphism, and the fixed vector space $\mathbb{R}^n$ matches the dimension of the tangent spaces $T_\lambda \Lambda$. Giving the most natural possible structures to make this generalization possible is the goal of the next section.

\section{An Extension of G\"unther's Approach to the Non-Flat Case} \label{extension}

My extension of G\"unther's theory begins with a fiber bundle $\epsilon : E \to M$ with total space $E$, base manifold $M$, projection map $ \epsilon : E \to M$, and standard fiber $Q$. The base manifold $M$ represents the space-time of our theory (typically represented by $\mathbb{R}^4$, but which we will assume is merely locally diffeomorphic to $\mathbb{R}^n$ for generality), while the space $Q \simeq E_x := \epsilon^{-1}(x)$ represents the space in which the physical field takes its values (typically $ \mathbb{R}^N $). The total space $E$, therefore, represents possible pairings between space-time values and field values; in other words, all possible field configurations over space-time. In physical terms, $E$ is the extended configuration space of our theory. Local fibered coordinates on $E$ are given by
$$\{ x^\alpha, \phi^I \} : E \to \mathbb{R}^{n+N} \mid e \mapsto x^\alpha e_\alpha + \phi^I e_I $$

The next step is to construct the vector bundle 
$$ V := V E \otimes_E T^* M \simeq \vec J E $$
where the tensor product is really fiberwise between $V_e E$ and $T^*_{\epsilon(e)} M$ \footnote{This slightly abusive notation is relatively common in the physics literature on covariant Hamiltonian field theory and serves to highlight what we take to be the base space of the resulting fiber bundle; see for example \cite{giachetta1999covariant}. A more thorough description of this bundle is given below.}. This vector bundle is isomorphic to the linearized first jet bundle over $E$; that is, the vector bundle upon which the ordinary jet bundle $J E$ is modeled. The first jet bundle $J E$ (an affine bundle over $E$) is the foundation for most discussions of covariant Hamiltonian field theory (see \cite{forger2003poisson} for notation and details). More specifically, the bundle $V$ is defined to be the bundle over $E$ with fibers $V_e E \otimes T^*_{\epsilon(e)} M $, where $VE$ is the vertical bundle over $E$ (see, for example, \cite{forger2003poisson} and \cite{gotay1998momentum}) with its fibers defined by $ V_e E := \{ u \in T_e E \mid T_e \epsilon (u) = 0 \} $, and $\epsilon(e)$ is the base space point over which $e \in E$ lies. We note, since it will be relevant later, that this bundle can be re-interpreted as a vector bundle over $M$ rather than $E$. \\

The bundle that forms the foundation for the geometric structures of this paper is really the dual of this bundle:
\begin{equation} P := V^*E \otimes_E TM \label{defP} \end{equation}
This bundle has standard fiber $P_e := V^*_e E \otimes T_{\epsilon(e)} M$. Like the previous vector bundle, it can be interpreted either as a bundle over $E$ with projection map $\pi : P \to E$ or as a bundle over $M$ with projection map $ \epsilon \circ \pi : P \to M $. We will make extensive use of both of these projections in defining the structures of our theory and their action on physical fields.

$P$ represents possible combinations of three things: 1) a point in space-time, 2) the values of all components of a physical field, and 3) a conjugate momentum value to each field value for each dimension of space-time. In physical terms, $P$ is the extended phase space of our theory. Local coordinates on $P$ that are compatible with a coordinate system $\{ x^\alpha, \phi^I \}$ on $E$ are given by 
$$\{ x^\alpha, \phi^I, \pi^\alpha_I \} : P \to \mathbb{R}^{n+N+nN} \mid p \mapsto x^\alpha e_a + \phi^I e_I + \pi^\alpha_I d \phi^I \otimes \pa{}{x^\alpha}$$
Note that this space is almost identical to the one considered by G\"unther in \cite{gunther1987polysymplectic}, except that 1) the base space $M$ (which I called $\Lambda$ in the section \ref{gunther's}) is not suppressed and 2) the poly-momenta take their vector-values in each $T_m M$, rather than a single copy of $\mathbb{R}^n$. These two seemingly innocuous changes will turn out to make a great deal of difference. \\

To define the tautological tensor $\theta : TM \to T^*P$ (a vector-valued one-form, as in section \ref{gunther's}) it is necessary to first consider an element of the space $T_p P \otimes T^*_{\epsilon \circ \pi(p)}M$. Given such an element $u = v \otimes \alpha$, with $v \in T_p P$ and $\alpha \in T_{\epsilon \circ \pi (p)}^*M$, the goal is to act upon it with the point $p$, considered as an element of the space $P_e$ -- that is, as a map $p : VP \otimes T^*M \to \mathbb{R}$ -- as in section \ref{gunther's}. Certainly the map $p$ cannot be applied to $u$ without modification. The first step, as before, is to apply $T_p \pi : T_p P \to T_{\pi(p)} E$ by extending it to the tensor product $TP \otimes T^*M$. I will take the liberty of calling this extended mapping $T_p \pi : T_p P \otimes T^*_{\epsilon \circ \pi(p)} M \to T_{\pi(p)} E \otimes T^*_{\epsilon \circ \pi(p)} M$ by the same name. (Since push-forwards by projection maps and identity maps on vector spaces are linear, the universal property of the tensor product \cite{birkhoff1965modern} guarantees that the map $ T \pi : T P \to TE $ can be uniquely extended to $TP \otimes T^*M$.) But the result is an element of $T_{\pi(p)} E \otimes T^*_{\epsilon \circ \pi(p)} M $, not $V_{\pi(p)}E \otimes T^*_{\epsilon \circ \pi(p)} M$ as desired, a natural consequence of the fact that I have not suppressed the parameter space $M$ as in section \ref{gunther's}. To map this element $u$ appropriately requires additional structure, namely a vertical projection (more abstractly called an Ehresmann connection) on TE. This is a vector bundle homomorphism (linear map) $V_E: TE \to VE$ such that 1) $VE = \text{Im} \ V_E$ and 2) $V_E^2 = V_E$. Given such a map, $\theta$ can be defined as follows:
\begin{definition} \label{tautologicaltensordefinition} Given an appropriately constructed fiber bundle $P$ as in eq.\ \eqref{defP} and a vertical projection operator $V_E : TE \to VE$ as defined above, the tautological tensor associated with $P$ and $V_E$ is defined pointwise by $$ \theta_p(u) := p \circ V_E \circ T \pi (u) $$ \end{definition}

Since any vertical projection operator is fiberwise linear, $\theta$ is naturally a linear map from $TP \otimes T^*M \to \mathbb{R}$ and therefore a section of $T^*P \otimes TM$. The exact coordinate representation of $\theta$ depends upon the vertical projection $V_E$. Choosing an arbitrary set of local fibered coordinates for $E$ gives $V_E : TE \to VE \mid u^\alpha \pa{}{x^\alpha} + u^I \pa{}{\phi^I} \mapsto (V^I_\alpha u^\alpha + u^I) \pa{}{\phi^I} $, so the result for $\theta$ is
$$ \theta_p = \pi^\alpha_I d \phi^I \otimes \pa{}{x^\alpha} + \pi^\alpha_I V^I_\beta dx^\beta \otimes \pa{}{x^\alpha}$$
where the $V^I_\alpha$ are the (arbitrary) coefficients of the vertical projection $V_E$. In the case that these coefficients are $0$ (as is natural in the case where $E = \mathbb{R}^n \times Q$) the result is exactly the same as in section \ref{gunther's}:
$$ \theta_p = \pi^\alpha_I d \phi^I \otimes \pa{}{x^\alpha} $$
Local coordinates in which $\theta = \pi^\alpha_I d \phi^I \otimes \pa{}{x^\alpha}$ are often called canonical coordinates\footnote{Though this terminology matches that of many works on symplectic geometry, it does not match \cite{gunther1987polysymplectic}. As this condition is very strong I will not make much use of it in the examples.}. Since a space need not admit any flat vertical projection, these coordinates are not guaranteed to exist even locally. \\

Since $\theta$ no longer takes its values in a single vector space but instead in the many tangent spaces $T_m M$, it does not have a canonical exterior derivative. Indeed, given its rather strange pedigree, defining any exterior derivative on $\theta$ is somewhat challenging. (It is not even a vector-valued exterior form in the usual sense, since our vector fields are sections of $TM$, rather than $P$ or $TP$). Based on the fact that the parameter space was suppressed throughout section \ref{gunther's}, we expect that the polysymplectic structure should be vertical in some sense. Indeed, something similar to G\"unther's basic approach can be made to work with a few modifications. 

Given any two vector fields $u, \, v \in \Gamma(P,TP)$ and any one-form $\beta \in \Gamma(M,T^*M)$, define the (non-tensorial) map $\omega$ via 
\begin{definition} \label{symplecticstructuredefinition} Given a tautological tensor $\theta$ as in definition \ref{tautologicaltensordefinition}, the associated symplectic structure $\omega$ is defined by
$$ \omega(u, v, \beta) := - d(\theta \lrcorner \beta) (u,v) $$
where $\lrcorner$ denotes the contraction of the contravariant part of the tensor on the left with the covariant components of the form on the right. \\
\end{definition}
More precisely, if we choose to represent $\theta(p) = \alpha \otimes m$, where $\alpha \in V^*_p P \subset T^*_p P$ (with $VP$ defined with the projection map $\epsilon \circ \pi$; i.e., as a bundle over $M$) and $m \in T_{\epsilon \circ \pi(p)} M$, then $\theta \lrcorner \beta (p) := \alpha(p) \beta_{\epsilon \circ \pi(p)} (m)$. The result is a linear map $\theta \lrcorner \beta : VP \subset TP \to \mathbb{R}$, and is therefore a section of $V^*P \subset T^*P$.

Though this is a well-defined map, it is not a tensor: the result of $d (\theta \lrcorner \beta)$ depends upon the particular one-form $\beta$, whereas a polysymplectic structure should be multi-linear. But looking at things in local coordinates makes it clear that defining $\omega$ only on vertical vector fields makes this structure unique; indeed, the domain of  $\theta \lrcorner \beta$ was already restricted to vertical vectors fields, so this restriction on $\omega$ is actually quite natural. These considerations lead to the following:
\begin{theorem} \label{symplectictensor}
When its domain is restricted to vertical vectors, the map $\omega$ of definition \ref{symplecticstructuredefinition} is actually a tensor, with coordinate representation $$ \omega = d \phi^I \wedge d \pi^\alpha_I \otimes \pa{}{x^\alpha} $$ in any local fibered coordinate system, independent of the vertical projection operator $V_E$ used to define $\theta$.
\end{theorem}
\begin{proof}
We use arbitrary local fibered coordinates to prove both assertions. In coordinates in which a point $p \in P = (x, \phi, \pi) = x^\alpha e_\alpha + \phi^I e+I + \pi^\alpha_I d \phi^I \otimes \pa{}{x^\alpha}$, the tautological tensor $\theta = \pi^\alpha_I d \phi^I \otimes \pa{}{x^\alpha} + \pi^\alpha_I V^I_\beta(x,e) dx^\beta \otimes \pa{}{x^\alpha}$, and our one form $\beta = \beta_\alpha(x) dx^\alpha$, we have
\begin{multline*}
- d (\theta \lrcorner \beta) = - d (\beta_\alpha \pi^\alpha_I d \phi^I + \beta_\alpha \pi^\alpha_I V^I_\beta dx^\beta) = \\\beta_\alpha d \phi^I \wedge d \pi^\alpha_I + \beta_\alpha V^I_\beta dx^\beta \wedge d \pi^\alpha_I + (\pi^\alpha_J \pa{\beta_\alpha}{x^\gamma} - \pi^\alpha_I \beta_\alpha \pa{V^I_\gamma}{\phi^J}) d \phi^J \wedge dx^\gamma + ( \pi^\alpha_I \pa{\beta_\alpha}{x^\gamma} V^I_\beta + \pi^\alpha_I \beta_\alpha \pa{V^I_\beta}{x^\gamma}) dx^\beta \wedge dx^\gamma
\end{multline*}
so inputting two vertical vector fields $u, \, v \in \Gamma(P,VP)$, where $u = u^I \pa{}{\phi^I} + u^\alpha_I \pa{}{\pi^\alpha_I}$ and $v = v^I \pa{}{\phi^I} + v^\alpha_I \pa{}{\pi^\alpha_I}$ gives
$$ - d (\theta \lrcorner \beta) (u,v) = \beta_\alpha (u^\alpha_I v^I - u^I v^\alpha_I) $$
which is multilinear in $\beta$, $u$, and $v$. Therefore, in any local fibered coordinates, the foregoing means that we can write the resulting tensor $\omega_p : V_p P \times V_p P \times T_{\pi \circ \epsilon(p)}^*M \to \mathbb{R}$ as 
$$ \omega = d \phi^I \wedge d \pi^\alpha_I \otimes \pa{}{x^\alpha} $$
\end{proof}

This is directly analogous to what was found in section \ref{gunther's}. Note that this local result is independent of the particular vertical projection $V_E$ used to define the tautological tensor of in definition \ref{tautologicaltensordefinition}. In the end, the apparent complication of this additional geometric structure does not affect the local physical result. \\

To make contact with physics, the last step is to introduce a Hamiltonian function $H : P \to \mathbb{R}$ that encodes the dynamics of the theory and stipulate how it is related to the geometric structures we have introduced. (Note that $H$ can now depend explicitly on the spacetime coordinates $x^\alpha$, since the base space $M$ is incorporated into $P$, unlike in section \ref{gunther's}.) It would be nice to proceed exactly as in section \ref{gunther's}: given a section $\gamma : M \to P$ of $P$, take the differential $T_m \gamma : T_m M \to T_{\gamma(m)} P$, contract it with the polysymplectic structure $\omega$, and demand that it match the exterior derivative of $H$. However, this does not work because the domain of $\omega$ is restricted to vertical vectors. Once again, the straightforward solution is to introduce a vertical projection $V_P$ -- this time on $TP$ instead of $TE$ -- and to extend it to the tensor product $TP \otimes T^*M$ as above. 
Given the polysymplectic structure defined above, physical solution sections $\gamma$ might be identified as those that satisfy 
\begin{equation} 
\omega ( V_P \circ T \gamma, v) = d H (v)
\label{hamiltonsfieldequationsgeometric}
\end{equation}
at every point $p \in \text{Im} \gamma$, for all vertical vectors fields $v \in \Gamma(P,VP)$. Here, $H : P \to \mathbb{R}$ is the covariant Hamiltonian function. Just as in section \ref{gunther's}, this function encodes the dynamical properties of the physical system under consideration. The exact coordinate representation of these sections now depends upon the vertical projection $V_P$. \\

With all this structure in place, it is worth noting that what looks like two new structures -- the vertical projections $V_E$ and $V_P$ -- need really only be one new structure. Every vertical projection $V_P$ on $TP$ determines a unique compatible vertical projection $V_E$ on $TE$ through the requirement that the following diagram commute:
\begin{center}
\begin{tikzcd}
TP \arrow[d, "T \pi"] \arrow[r, "V_P"] & \arrow[d, "T \pi"] VP \\
TE \arrow[r, "V_E"] & VE \\
\end{tikzcd}
\end{center}
In compatible local fibered coordinates on $P$ and $E$ in which $V_P : TP \to VP \mid u^\alpha \pa{}{x^\alpha} + u^I \pa{}{\phi^I} + u^\alpha_I \pa{}{\pi^\alpha_I} \mapsto (V^I_\alpha u^\alpha + u^I) \pa{}{\phi^I} + (V^\alpha_{I \beta} u^\beta + u^\alpha_I) \pa{}{\pi^\alpha_I}$, this just amounts to the requirement that $V^{I}_{\alpha \ E} = V^{I}_{\alpha \ P}$. \\

Keeping in mind the fact that there is only a single vertical projection $V_P = V$ to be specified and choosing local fibered coordinates in which $V : TP \to VP \mid u^\alpha \pa{}{x^\alpha} + u^I \pa{}{\phi^I} + u^\alpha_I \pa{}{\pi^\alpha_I} \mapsto (V^I_\alpha u^\alpha + u^I) \pa{}{\phi^I} + (V^\alpha_{I \beta} u^\beta + u^\alpha_I) \pa{}{\pi^\alpha_I}$, eq.\ \eqref{hamiltonsfieldequationsgeometric} amounts to the local fibered coordinate equations
\begin{equation} \pa{\gamma^I}{x^\alpha} + V^I_\alpha = \pa{H}{\pi^\alpha_I} \label{localhfe1} \end{equation}
\begin{equation} \pa{\gamma^\alpha_I}{x^\alpha} + V^\alpha_{I \alpha} = - \pa{H}{\phi^I} \label{localhfe2} \end{equation}
so in the case that $E = \mathbb{R}^n \times Q$ and the vertical projection is the natural flat one the result is exactly the same as in section \ref{gunther's}:
$$ \pa{\gamma^I}{x^\alpha} = \pa{H}{\pi^\alpha_I} $$
$$ \pa{\gamma^\alpha_I}{x^\alpha} = - \pa{H}{\phi^I} $$
which are Hamilton's field equations.

However, if the vertical projection is not naturally flat -- and generically it cannot be -- then we do not recover Hamilton's field equations unless we use a modified Hamiltonian function that now depends upon the vertical projection $V$ as well as the system under consideration:
$$ H' := H - \pi^\alpha_I V^I_\alpha + \phi^I V^\alpha_{I \alpha} $$
This seems a very artificial way to produce the physically correct field equations. Moreover, it is not generically possible to perform this modifiction globally. Luckily, there is another, more elegant way to circumvent this problem in the cases of greatest physical interest. I will describe this modification to G\"unther's approach in section \ref{modifications}. But already we have enough of a foundation in place to discuss generalized Poisson structures. \\

\section{Generalized Poisson Structures} \label{poisson}
Poisson brackets are not strictly necessary for doing covariant Hamiltonian field theory, as evidenced by their absence from most versions \cite{forger2003poisson}. However, as Hamiltonian vector fields and Poisson brackets play an essential role in the most straightforward geometric quantization methods for Hamiltonian particle theory, they may be essential to formulating a quantum counterpart to covariant Hamiltonian field theory. Therefore, I take the attitude that the formulation of a field theoretic counterpart to the Poisson bracket of Hamiltonian particle theory is an important component of any covariant Hamiltonian field theory. 

To define a second symplectic tensor that is the counterpart to the Poisson structure in geometric Hamiltonian particle theory, I begin by noting that it is possible to associate with any function $f$ on $P$ a \emph{family} of tensors $\mathscr{S}_f$, where each section $s_f \in \mathscr{S}_f \subset \Gamma(P, VP \otimes T^*M) $ must satisfy:
$$ \omega(s_f, v) = v^\alpha_I s^I_\alpha - v^I s^\alpha_{I \alpha} = df(v) $$
and the relation is expected to hold for all vertical vector fields $v : P \to VP$. In coordinates, these sections have components obeying the relations
$$ s^I_\alpha = \pa{f}{\pi^\alpha_I}, \quad s^\alpha_{I \alpha} = - \pa{f}{\phi^I}$$
There is a family of sections rather than a single section associated with each function $f$ because the second relation specifies only the trace of the second set of coordinate functions $s^\alpha_{I \beta}$, rather than specifying every coordinate function uniquely. In local fibered coordinates, the sections $s_f$ look like:
$$ s_f = \pa{f}{\pi^\alpha_I} \ dx^\alpha \otimes \pa{}{\phi^I} - \pa{f}{\phi^I} \ dx^\alpha \otimes \pa{}{\pi^\alpha_I} + s^{\quad \ \beta}_{TF \, I \alpha} \ dx^\alpha \otimes \pa{}{\pi^\beta_I} $$
where the components $s_{TF}$ (TF stands for ``trace-free'') are arbitrary other than that they must obey the condition $s^{\quad \ \alpha}_{TF \, I \alpha} = 0$.

Though all of this follows more-or-less naturally from considerations of verticality and the polysymplectic structure, it would be more useful to have a \emph{unique} section associated with each function on the phase space $P$. This we can do by the following:

\begin{theorem} \label{poissonsection}
For every function $f \in C^\infty(P)$, it is possible to define a unique section $\sigma_f \in \mathscr{S}_f$ by requiring that, for all functions $g \in C^\infty(P)$, we have
$$ \sigma_f(dg) = - \sigma_g(df) $$
\end{theorem}
\begin{proof}
Let us choose local coordinates in which $\sigma_f = \pa{f}{\pi^\alpha_I} \ \pa{}{\phi^I} \otimes dx^\alpha + F^{\beta}_{J \alpha} \ \pa{}{\pi^\beta_J} \otimes dx^\alpha$ and $\sigma_g = \pa{g}{\pi^\alpha_I} \ \pa{}{\phi^I} \otimes dx^\alpha + G^{\beta}_{J \alpha} \ \pa{}{\pi^\beta_J} \otimes dx^\alpha$ and let us define at each point $p \in P$ new components $\tilde F^\beta_{J \alpha} = F^\beta_{J \alpha} + \delta^\beta_\alpha \pa{f}{\phi^J}$ and $\tilde G^\beta_{J \alpha} = G^\beta_{J \alpha} + \delta^\beta_\alpha \pa{g}{\phi^J}$, where the $\delta^\beta_\alpha$ are the components of the identity matrix of dimension $\dim M$. In these local fibered coordinates, at each point $p \in P$ the requirement $ \sigma_f(dg) = - \sigma_g(df) $ reads
\begin{multline*}
\bigg( \pa{f}{\pi^\alpha_I} \ \pa{g}{\phi^I} - \delta^\beta_\alpha \pa{f}{\phi^J}\ \pa{g}{\pi^\beta_J} + \tilde F^{\beta}_{J \alpha} \ \pa{g}{\pi^\beta_J}
 + \pa{g}{\pi^\alpha_I} \ \pa{f}{\phi^I} - \delta^\beta_\alpha \pa{g}{\phi^J}\ \pa{f}{\pi^\beta_J} + \tilde G^{\beta}_{J \alpha} \ \pa{f}{\pi^\beta_J} \bigg) dx^\alpha = 0
\end{multline*}
Relabeling indices, eliminating delta-functions, and canceling like terms yields the requirement
$$ \bigg( \tilde F^{\beta}_{I \alpha} \ \pa{g}{\pi^\beta_I} + \tilde G^{\beta}_{I \alpha} \ \pa{f}{\pi^\beta_I} \bigg) dx^\alpha = 0 $$
Since the relation is to hold for all $g \in C^\infty(P)$ and for each component $\alpha$, and since $\tilde F^{\beta}_{I \alpha}$ and $\pa{f}{\pi^\beta_I}$ are independent, the requirement can be reformulated as two independent conditions
$$ \tilde F^{\beta}_{I \alpha} \ \pa{g}{\pi^\beta_I} = 0 = \tilde G^{\beta}_{I \alpha} \ \pa{f}{\pi^\beta_I} $$
for each $\alpha$. The first condition then requires that  
$$ \tilde F^{\beta}_{I \alpha} = 0 $$
regardless of the derivatives $\pa{f}{\pi^\beta_I}$. (If the derivatives $\pa{f}{\pi^\beta_I}$ are all non-vanishing, we must have $\tilde G^{\beta}_{I \alpha} = 0 $ also, but this is consistent with reversing the roles of $f$ and $g$.) Therefore, each $\sigma_f$ is uniquely determined, point-by-point, to be
$$ \sigma_f = \pa{f}{\pi^\alpha_I} \ \pa{}{\phi^I} \otimes dx^\alpha - \delta^\beta_\alpha \pa{f}{\phi^J} \ \pa{}{\pi^\beta_J} \otimes dx^\alpha $$
Since every component is uniquely determined at every point, the section $\sigma_f$ is unique.
\end{proof}

Considering these more tractable sections leads to the following:

\begin{theorem} \label{poissonstructure}
At each point $p \in P$, there is a unique tensor $\Pi_p : T_{\epsilon \circ \pi(p)} M \times T^*_p P  \times T_p^*P \to \mathbb{R}$ such that, for all functions $f \in C^\infty(P)$, $\Pi_p(-,df,-) = \sigma_f$.
\end{theorem}
\begin{proof}
In local fibered coordinates, a general tensor $\Pi \in T^*_{\epsilon \circ \pi(p)} M \otimes T_p P \otimes T_p P $ can be written
\begin{multline*} \Pi = \Pi^{\beta \gamma}_{\alpha} dx^\alpha \otimes \pa{}{x^\beta} \otimes \pa{}{x^\gamma} + \Pi^{\beta I}_{\alpha} dx^\alpha \otimes \pa{}{x^\beta} \otimes \pa{}{\phi^I} + \Pi^{\beta \gamma}_{\alpha I} dx^\alpha \otimes \pa{}{x^\beta} \otimes \pa{}{\pi^\gamma_I} \\
 +  \Pi^{I \beta}_\alpha dx^\alpha \otimes \pa{}{\phi^I} \otimes \pa{}{x^\beta}+ \Pi^{I J}_{\alpha} dx^\alpha \otimes \pa{}{\phi^I} \otimes \pa{}{\phi^J} +  \Pi^{I \beta}_{\alpha J} dx^\alpha \otimes \pa{}{\phi^I} \otimes \pa{}{\pi^\beta_J} \\
 + \Pi^{\beta \gamma}_{ \alpha I} dx^\alpha \otimes \pa{}{\pi^\beta_I} \otimes \pa{}{x^\gamma} + \Pi^{\beta J}_{\alpha I} dx^\alpha \otimes \pa{}{\pi^\beta_I} \otimes \pa{}{\phi^J} + \Pi^{\beta \gamma}_{\alpha I J} dx^\alpha \otimes \pa{}{\pi^\beta_I} \otimes \pa{}{\pi^\gamma_J} \end{multline*}
Contracting this tensor with $df$ gives
\begin{multline*} \Pi(-,df,-) = \Pi^{\beta \gamma}_{\alpha} \pa{f}{x^\beta} dx^\alpha \otimes \pa{}{x^\gamma} + \Pi^{\beta I}_{\alpha} \pa{f}{x^\beta} dx^\alpha \otimes \pa{}{\phi^I} + \Pi^{\beta \gamma}_{\alpha I} \pa{f}{x^\beta} dx^\alpha \otimes \pa{}{\pi^\gamma_I} \\
 +  \Pi^{I \beta}_\alpha \pa{f}{\phi^I} dx^\alpha \otimes \pa{}{x^\beta}+ \Pi^{I J}_{\alpha} \pa{f}{\phi^I} dx^\alpha \otimes \pa{}{\phi^J} +  \Pi^{I \beta}_{\alpha J} \pa{f}{\phi^I} dx^\alpha \otimes \pa{}{\pi^\beta_J} \\
 + \Pi^{\beta \gamma}_{\alpha I} \pa{f}{\pi^\beta_I} dx^\alpha \otimes \pa{}{x^\gamma} + \Pi^{\beta J}_{\alpha I} \pa{f}{\pi^\beta_I} dx^\alpha \otimes \pa{}{\phi^J} + \Pi^{\beta \gamma}_{\alpha I J} \pa{f}{\pi^\beta_I} dx^\alpha \otimes \pa{}{\pi^\gamma_J} \end{multline*}
Relabeling indices and comparing to $\sigma_f$ then requires that:
\begin{multline*}  \left( \Pi^{\gamma \beta}_\alpha \pa{f}{x^\gamma} + \Pi^{I \beta}_\alpha \pa{f}{\phi^I} + \Pi^{\gamma \beta}_{\alpha I} \pa{f}{\pi^\gamma_I} \right) dx^\alpha \otimes \pa{}{x^\beta} \\
 + \left( \Pi^{\beta I}_{\alpha} \pa{f}{x^\beta} + \Pi^{J I}_{\alpha} \pa{f}{\phi^J} + \Pi^{\beta I}_{\alpha J} \pa{f}{\pi^\beta_J} \right) dx^\alpha \otimes \pa{}{\phi^I} \\
+  \left( \Pi^{\gamma \beta}_{\alpha I} \pa{f}{x^\gamma} + \Pi^{J \beta}_{\alpha I} \pa{f}{\phi^J} + \Pi^{\gamma \beta}_{\alpha J I} \pa{f}{\pi^\gamma_J} \right) dx^\alpha \otimes \pa{}{\pi^\beta_I} \\
 = 0 \ dx^\alpha \otimes \pa{}{x^\beta} + \pa{f}{\pi^\alpha_I} dx^\alpha \otimes \pa{}{\phi^I} - \pa{f}{\phi^I} \delta^\beta_\alpha dx^\alpha \otimes \pa{}{\pi^\beta_I} \end{multline*}
Since $f$ is arbitrary, the first line requires that $\Pi^{\gamma \beta}_{\alpha} =  \Pi^{I \beta}_{\alpha} = \Pi^{\gamma \beta}_{\alpha I} = 0 $, the second further requires that $\Pi^{\beta I}_\alpha = \Pi^{J I}_{\alpha} =0 $ and that $ \Pi^{\beta I}_{\alpha J} = \delta^I_J \delta^\beta_\alpha $, and the third further requires that $\Pi^{\gamma \beta}_{\alpha I} = \Pi^{\gamma \beta}_{\alpha J I} = 0$ and that $ \Pi^{J \beta}_{\alpha I} = - \delta^J_I \delta^\beta_\alpha $.
Putting all this together gives
\begin{equation} \Pi_p = \left( - \pa{}{\phi^I} \otimes \pa{}{\pi^\alpha_I} + \pa{}{\pi^\alpha_I} \otimes \pa{}{\phi^I} \right) \otimes dx^\alpha = - \pa{}{\phi^I} \wedge \pa{}{\pi^\alpha_I} \otimes dx^\alpha \label{poissonstructurecoordinates} \end{equation}
with all tensor components uniquely defined.
\end{proof}

The primary significance of $\Pi$ is that it allows us to define a generalized Poisson bracket in a natural and invariant manner: for any two functions $f, \, g \in C^\infty(P)$, we define their generalized Poisson bracket $\{ f, \, g \}$ by
$$ \{ f, \, g \} := \Pi(-, df, dg) $$
In coordinates, this reads
$$ \{ f, \, g \} = \left( \pa{f}{\phi^I} \pa{g}{\pi^\alpha_I} - \pa{f}{\pi^\alpha_I} \pa{g}{\phi^I} \right) \, dx^\alpha $$
This means that the generalized Poisson bracket of two functions is represented in this theory by a one-form over $T^*M$. It arises more-or-less naturally in this geometric setting, with all the natural transformation properties of the generalized Poisson structure of \cite{struckmeier2008covariant}. However, it is important to note that this generalized Poisson bracket is not a true Poisson bracket, as it map pairs of functions to one-forms rather than functions, and does not satisfy the ordinary Jacobi identity.

It is worth noting that physically desirable generalized Poisson structures such as these are difficult or impossible to generate in many covariant Hamiltonian field theories \cite{forger2003poisson}. The existence of such a structure is therefore one of the major successes of the theory.

\section{Further Modifications} \label{modifications}
Though the geometric structures of sections \ref{extension} and \ref{poisson} hew very closely to G\"unther's original approach, the need to define a vertical projection operator on $TP$ in order to use eq.\ \eqref{hamiltonsfieldequationsgeometric} is a substantial challenge: it is often not obvious how best to define the correct vertical projection operator, and the result depends very much upon our choice. However, though equation eq.\ \eqref{hamiltonsfieldequationsgeometric} mirrors very closely G\"unther's original approach, it is not the only way to use the symplectic structure to produce meaingful field equations. In this section I will show how use physically meaningful connections on the spaces underlying the extended phase space $P$ to construct from any prospective solution section $\gamma : M \to P$ a vertically-valued section that can be fed directly into the symplectic structure of theorem \ref{symplectictensor}. 

This alternative construction is similar to the idea of letting each prospective solution section define its own vertical projection operator. If we are willing to slightly bend the prescription of eq.\ \eqref{hamiltonsfieldequationsgeometric}, we will find that an approach in which each section defines its own vertically-valued tensor that can be fed directly into the symplectic structure can be made to work. 

This approach can be broken down into three steps:
\begin{enumerate}
\item Define a connection on sections $\gamma : M \to P$
\item Vertically lift the gauge covariant derivative of a prospective solution section $\gamma : M \to P$ to $VP \otimes TM$
\item Contract this vertically lifted gauage covariant exterior derivative with the symplectic structure to produce the appropriate field equations
\end{enumerate}
None of these steps is standard, so I will discuss each one in order. 

First, we wish to define a connection on sections of $\epsilon \circ \pi : P \to M$. We will take as a given that there exists a Levi-Civita connection $\nabla_{TM}$ uniquely defined by the metric tensor $g$ naturally occuring in the definition of the Lagrangian density $\mathscr{L}$; see section \ref{applications} for examples. We will then also need to define connections $\nabla_E$ on sections of $\epsilon: E \to M$ and $\nabla_{V^*E}$ on $\epsilon \circ \zeta^* : V^*E \to M$ in order to have a well defined connection on $P$. (In fact, we will usually construct a connection $\nabla_{VE}$ -- which will naturally induce the connection $\nabla_{V^*E}$ -- from a connection on $E$ via the vertical lift.) First, let us consider the following: 
\begin{proposition} \label{liftedsection}
Every section $\gamma$ of $\epsilon \circ \zeta : VE \to M$ can be uniquely identified with a pair of sections $\gamma_1, \gamma_2 : M \to E$ via
$$ \gamma(x) = \text{vl}_{\gamma_1(x)} \gamma_2(x) $$
\end{proposition}
\begin{proof}
In coordinates in which $\gamma = x^\alpha e_\alpha + \gamma^I e_I + \gamma^J \pa{}{\phi^J}$, $\gamma_1 =  x^\alpha e_\alpha + \gamma_1^I e_I $, and $\gamma_2 = x^\alpha e_\alpha + \gamma_2^I e_I $ we have
$$ \gamma =  x^\alpha e_\alpha + \gamma_1^I e_I  + \gamma_2^I \pa{}{\phi^I}$$
Since the sections $\gamma_1, \gamma_2$ are globally defined, this identification is also global. The fact that there is only one pair of sections $\gamma_1, \gamma_2$ that we can associate with a given section $\gamma$ shows that this identification is unique.
\end{proof}
This ability to uniquely associate sections of $VE$ with pairs of sections of $E$ allows us to lift connections on $E$ to connections on $VE$. Let us make this precise with the following
\begin{proposition} \label{liftedconnection}
Given a connection $\nabla$ on sections of $E$, we can define a connection $\nabla_{VE}$ on sections of $\epsilon \circ \zeta : VE \to M$ via
$$ \nabla_{VE,X} \gamma := \text{vl}_{\gamma_1} \nabla_{X} \gamma_2 $$
where $\gamma_1, \gamma_2$ are the sections of $E$ associated with $\gamma$ from proposition \ref{liftedsection} and the equality is to hold for all vector fields $X : M \to TM$.
\end{proposition}
\begin{proof}
The most importnat thing to check is that the result is in fact a section of $\epsilon \circ \zeta : VE \to M$. Looking at things in local fibered coordinates makes this clear:
$$ \text{vl}_{\gamma_1} \nabla_{X} \gamma_2 = x^\alpha e_\alpha + \gamma^I_1(x) e_I + X^\alpha \left( \pa{\gamma^I_2(x)}{x^\alpha} + A^I_{J \alpha} \gamma^J_2(x) \right) \pa{}{\phi^I} $$
which is a legitimate section of the bundle in question, globally well defined when the $A$s are the connection potentials associated with the original connection $\nabla$.
\end{proof} In coordinates, this procedure amounts to defining a new connection on $VE$ with the same connection potentials as the starting connection $\nabla$ on $E$. It will be useful to note that this connection on $VE$ naturally induces a connection $\nabla_{V^*E}$ on $\epsilon \circ \zeta^* : V^*E \to M$ via a standard construction through the requirement that 
$$ \nabla(\gamma \gamma^*) =: \gamma^* \nabla_{VE} \gamma + \gamma \nabla_{V^*E} \gamma^* = d (\gamma \gamma^*) $$
for all sections $\gamma : M \to VE$ and $\gamma^*: M \to V^*E$ at all points $x \in M$. The result is that the connection potentials of $\nabla_{V^*E}$ differ from those of $\nabla_{VE}$ by a minus sign. \\

With a few modifications, this framework can be extended to the situation of interest in this paper. First, in the spirit of proposition \ref{liftedsection}, an arbitrary section $\gamma$ of $\epsilon \circ \pi : P \to M$ can be identified with two sections $\gamma_1 : M \to E$ and $\gamma_2: M \to V^*E$ that are compatible in the sense that $\gamma_2(x) \in V^*_{\gamma_1(x)} E$ for all $x \in M$, as well as a vector field $X : M \to TM$:
$$ \gamma(x) = \gamma_1(x) + \gamma_2(x) \otimes X(x) $$
The identification is no longer unique because of the tensor product in the second term, but all results are indepedent of the particular way in which we choose $\gamma_2$ and $X$. 

Similarly, in the spirit of proposition \ref{liftedconnection}, a connection $\nabla_P$ on the space of such sections can be constructed from connections $\nabla_1$ on $E$, $\nabla_2$ on $V^*E$ (potentially -- but not necessarily -- the dual of the one on $VE$ naturally associated with $\nabla_1$), and $\nabla_{TM}$ on $TM$ (invariably the appropriate Levi-Civita connection) via
\begin{equation} \nabla_{P,Y} (\gamma) := \nabla_{1,Y} \gamma_1 + ( \nabla_2 \otimes \nabla_{TM} )_Y ( \gamma_2 \otimes X) \label{connectiononP} \end{equation}
where $\nabla_2 \otimes \nabla_{TM}$ is the standard connection induced on the tensor product space $V^*E \otimes_M TM$ by the connections $\nabla_2$ and $\nabla_{TM}$, and $\{\gamma_1, \gamma_2, X\}$ are three appropriate sections identified with $\gamma$ via the construction of the preceeding paragraph. Note that all three sections are subject to covariant differentiation, unlike in proposition \ref{liftedconnection}, but the result is still a globally well-defined section of the appropriate bundle. In coordinates in which $\gamma = (x^\alpha , \gamma^I_1 , \gamma_{I,2} X^\alpha) $ this amounts to
$$\nabla_{P,Y} (\gamma) = Y^\alpha \left( \pa{\gamma^I_1}{x^\alpha} + \gamma^J A^I_{J \alpha,1} \right) e_I + Y^\alpha \left( \pa{(\gamma_{I,2} X^\beta) }{x^\alpha} + X^\beta \gamma_{J,2} A^J_{I \alpha,2} + \gamma_{I,2} X^\delta \Gamma^\beta_{\delta \alpha} \right) d \phi^I \otimes \pa{}{x^\beta}$$
where the $A_1$ are the connection potentials of $\nabla_1$, the $A_2$ are the connections potentials of $\nabla_2$, and the $\Gamma$ are the potentials of the Levi-Civita connection. 

The resulting connection is naturally a map $\nabla_P : \Gamma(M,TM) \times \Gamma(M,P) \to \Gamma(M,P)$. It can be vertically lifted exactly as in proposition \ref{liftedconnection} to produce a new map $\nabla_{VP} : \Gamma(M,TM) \times \Gamma(M,VP) \to \Gamma(M,VP)$. Using exactly the construction of proposition \ref{liftedsection}, we can identify a section $\gamma : M \to VP$ with two sections $\gamma_1, \gamma_2 : M \to P$ via 
$$\gamma(x) = \text{vl}_{\gamma_1(x)} \gamma_2(x)$$
so that we can construct the new connection $\nabla_{VP}$ by
$$\nabla_{VP,X}(\gamma) := \text{vl}_{\gamma_1(x)} (\nabla_P \gamma_2)$$
In coordinates, this is a new connection on $VP$ with the same potentials as $\nabla_P$. 

The key obsevation now is that the image of $\nabla_{VP}$ lies in $VP \otimes T^*M$, putting it squarely in the domain of the symplectic tensor $\omega$ without the need to define any vertical projection operator at all. In other words, if we are willing to stretch a bit the expectations from \cite{gunther1987polysymplectic} that led to eq.\ \eqref{hamiltonsfieldequationsgeometric}, we can consider solution sections $\gamma : M \to P$ to be those which obey 
\begin{equation} \omega \left( \nabla_{VP} ( \text{vl}_\gamma \gamma), v \right) = dH(v) \label{alternativehamiltons} \end{equation}
for all vertical vector fields $v : P \to VP$, with $\nabla_{VP}$ as constructed above and $H$ an appropriate Hamiltonian function. In coordinates in which $\gamma = (x^\alpha , \phi^I = \gamma^I_1 , \pi^\alpha_I = \gamma_{I,2} X^\alpha) $, this is equivalent to
\begin{equation} \left( \pa{\phi^I}{x^\alpha} + \phi^J A^I_{J \alpha,1} \right) d\pi^\alpha_I + \left( \pa{\pi^\alpha_I }{x^\alpha} + \pi^\alpha_J A^J_{I \alpha,2} + \pi^\alpha_{I} \Gamma^\beta_{\alpha \beta} \right) d \phi^I = \pa{H}{\phi^I} d\phi^I + \pa{H}{\pi^\alpha_I} d\pi^\alpha_I \end{equation}
With judicious choices of Hamiltonians and connections, we will find this prescription to be physically correct below. 

\section{Applications} \label{applications}

Note that Hamilton's equations never constrain all the derivatives $\pa{\gamma^\alpha_I}{x^\beta}$, only a particular sum of them (as is necessary to avoid over-determining the solution sections). Taken together, they serve to identify physically realizable field configurations (and their conjugate momenta) from un-physical configurations. The first of Hamilton's equations usually reiterates the relationship between the conjugate momentum coordinates and the derivatives of the field-value part of the solution section \footnote{In the case that one begins with the Hamiltonian $H$ directly (that is, without defining a Lagragian density, defining the conjugate momenta, or performing a Legendre transformation), this first equation serves to \textbf{define} the conjugate momenta in terms of derivatives of the field-value part of the solution section.}, while the second mimics the Euler-Lagrange equation of motion of the system when the constraints imposed by the first equation are satisfied. The specific field configuration taken by a physical system is then determined by imposing appropriate initial conditions and solving the system of partial differential equations. \cite{struckmeier2008covariant} provides the details of how these field equations apply to a wide range of examples of physical interest. As an indication of how the abstract geometric framework described above pertains to specific physical fields -- and in particular how the choices needed to use eq.\ \eqref{alternativehamiltons} are to be made -- I consider a few simple examples here. Perhaps more importantly, I will show that the approach of section \ref{modifications} allows us to side-step the complications of adding a vertical projection to the ordinary covariant Hamiltonian approach.

\subsection{Scalar Field Theory} \label{scalarfield}
The first and simplest example of a classical field is a single, real-valued (i.e., uncharged) scalar field. I will now examine how this simplest case fits into the framework of the previous sections.

In the most basic case, the extended configuration space for the scalar field is $E = \mathbb{R}^4 \times \mathbb{R}$. It carries global fibered coordinates $\{ x^\alpha, \phi \}$ such that a point $e \in E$ is given by $e = x^\alpha e_\alpha + \phi \ e_\phi$. The appropriate extended phases space is then $P := V^*E \otimes_E TM = \mathbb{R}^4 \times \mathbb{R} \times (\mathbb{R} \otimes \mathbb{R}^4)$ with global fibered coordinates $\{ x^\alpha, \phi, \pi^\alpha \}$ such that a point $p \in P$ is given by $e =  x^\alpha e_\alpha + \phi \ e_\phi + \pi^\alpha d \phi \otimes \pa{}{x^\alpha}$. The key symplectic structures are given in these global fibered coordinates by $\theta = \pi^\alpha d \phi \otimes \pa{}{x^\alpha}$, $\omega = d\pi^\alpha \wedge d \phi \otimes \pa{}{x^\alpha}$, and (though not necessary for the present analysis) $\Pi = - \pa{}{\phi} \wedge \pa{}{\pi^\alpha} \otimes dx^\alpha$. Note that I have made use of the natural (flat) vertical projection in defining $\theta$, and will make use of it again below to derive Hamilton's field equations. Given the standard Klein-Gordon Hamiltonian $H = \frac{1}{2} \eta_{\alpha \beta} \pi^\alpha \pi^\beta + \frac{1}{2} m^2 \phi^2$ (where $\eta$ is the Minkowski metric in the ``mostly-minus" or ``West coast" metric signature convention) and the fact that there is a natural (flat) vertical projection on $P$, eq.\ \eqref{hamiltonsfieldequationsgeometric} implies that physically realizable sections $\gamma$ must obey
$$ \pa{\phi}{x^\alpha} = \pa{H}{\pi^\alpha} = \eta_{\alpha \beta} \pi^\beta $$
and
$$ \pa{\pi^\alpha}{x^\alpha} = - \pa{H}{\phi} = - m^2 \phi $$
Solving the first equation for $\pi^\alpha$, noting that $\pa{\eta^{\alpha \beta}}{x^\gamma} = 0$, and substituting into the second gives
\begin{equation} \label{kgeflat} \eta^{\alpha \beta} \pa{^2 \phi}{x^\alpha x^\beta} + m^2 \phi = 0 \end{equation}
which is the standard, flat space Klein-Gordon equation. \\

As a step toward more interesting results, consider a single, real-valued scalar field over a flat space-time (so that is still $E = \mathbb{R}^4 \times \mathbb{R}$), but with a non-flat metric; this is the situation usually encountered in general relativity. In this case, the flat vertical projection is still available, but it is important to be careful about deriving the correct Hamiltonian from the Lagrangian formulation. If we think of the classical action as $S = \int \mathscr{L} \ d^4 x$, picking off $\mathscr{L} = \frac{\sqrt {- \det g}}{2} (g^{\alpha \beta} \pa{\phi}{x^\alpha} \pa{\phi}{x^\beta} - m^2 \phi^2)$ as everything in the integrand except the top form $d^4 x$, then the appropriate covariant conjugate momenta are
$$ p^\alpha := \pa{\mathscr{L}}{\pa{\phi}{x^\alpha}} = \sqrt {- \det g} \ g^{\alpha \beta} \ \pa{\phi}{x^\beta} $$
so that 
$$H = p^\alpha \pa{\phi}{x^\alpha} - \mathscr{L} = \frac{1}{2 \sqrt{- \det g} } g_{\alpha \beta} p^\alpha p^\beta +  \frac{ \sqrt{- \det g}}{2} m^2 \phi^2 $$
and Hamilton's field equations read
$$ \pa{\phi}{x^\alpha} = \pa{H}{p^\alpha} =\frac{1}{\sqrt{- \det g}} \ g_{\alpha \beta} p^\beta$$
and
$$ \pa{p^\alpha}{x^\alpha} = - \pa{H}{\phi} = - {\sqrt{- \det g} \ m^2 \phi} $$
Solving the first equation for $p^\alpha$ and substituting it into the second give
\begin{equation} 0 = m^2 \phi + \frac{1}{\sqrt{- \det g} } \pa{}{x^\alpha} ( \sqrt{- \det g} \ g^{\alpha \beta} \pa{\phi}{x^\beta}) = g^{\alpha \beta} \nabla_\alpha \nabla_\beta \phi + m^2 \phi \label{kge} \end{equation}
which is indeed the appropriate equation of motion for a single scalar field in a curved space-time. \\

But what happens when the underlying space-time is not $\mathbb{R}^4$ and one cannot access the flat vertical projection? In this formalism there are two options. One is to begin with the same Lagrangian density as in the previous case, perform the appropriate Legendre transformation, choose an arbitrary vertical projection on the underlying space to find Hamilton's equations, then finally alter the Hamiltonian function as noted just after eq.\ \eqref{hamiltonsfieldequationsgeometric} of section \ref{extension}. However, as the alternative outlined at the end of section \ref{extension} is more elegant, naturally global, and manifestly coordinate invariant, I will outline that method instead.

To determine the appropriate Hamiltonian $H$, let us begin again with the action, but in this case let us write it as:
$$ S = \int \text{vol} \ \mathscr{L} = \int \sqrt{- \det g} \ d^4 x \ \left( \frac{1}{2} g^{\mu \nu} \pa{\phi}{x^\mu} \pa{\phi}{x^\nu} - \frac{1}{2} m^2 \phi^2 \right)$$
Rather than use the ordinary separation of $\mathscr{L}$ from $d^4 x$ used successfully above, let us consider what happens when we single out the Lagrangian differently, being careful to isolate $\mathscr{L}$ not just from the top form $d^4 x$ but from the volume form $\text{vol} = \sqrt{- \det g} \ d^4 x $. This will have important ramifications later. 

The covariant momenta are then defined as the appropriate partial derivatives of the Lagrangian function $\mathscr{L}$:
$$ p^\mu := \pa{\mathscr{L}}{\pa{\phi}{x^\mu}} = g^{\mu \nu} \pa{\phi}{x^\nu}$$
so that the standard construction for the Hamiltonian gives
$$  H := p^\mu \pa{\phi}{x^\mu} - \mathscr{L} =  \frac{1}{2} g_{\mu \nu} p^\mu p^\nu + \frac{1}{2} m^2 \phi^2$$
Given local coordinates in which
$$ p \in P = x^\mu e_\mu + \phi e_\phi + p^\mu d \phi \otimes \pa{}{x^\mu} $$
we have
$$ \omega = d \phi \wedge d p^\mu \otimes \pa{}{x^\mu} $$
To make use of eq.\ \eqref{alternativehamiltons}, we shall use the natural connection $d$ on $E$ and $V^*E$ and the Levi-Civita connection $\nabla$ on $TM$ to define the connection $\nabla_P$ which in turn defines $\nabla_{VP}$. In coordinates in which we have a section $\gamma : M \to P \mid x^\mu e_\mu \mapsto \phi(x^\mu) e_\phi + p^\mu(x^\mu) d \phi \otimes \pa{}{x^\mu}$, we have
$$ \text{vl}_\gamma(\gamma) = x^\mu e_\mu + \phi e_\phi + p^\mu d \phi \otimes \pa{}{x^\mu} + \phi \pa{}{\phi} + p^\mu d \phi \otimes \pa{}{p^\mu} $$
so that
$$ \nabla_{VP} ( \text{vl}_\gamma(\gamma) ) = \left( \pa{\phi}{x^\mu} + 0 \right) \pa{}{\phi} \otimes dx^\mu + \left(\pa{p^\nu}{x^\mu} + \gamma^\xi \Gamma^\nu_{\xi \mu}  \right) \pa{}{p^\nu} \otimes dx^\mu $$
Eq.\ \eqref{alternativehamiltons} then gives
$$ \left( \pa{\phi}{x^\mu} + 0  \right) dp^\mu -  \left(\pa{p^\mu}{x^\mu} + \gamma^\xi \Gamma^\mu_{\xi \mu}  \right) d \phi = g_{\mu \nu} p^\nu \ d p^\mu + m^2 \phi \ d\phi $$
which is equivalent to 
$$ p^\mu = g^{\mu \nu} \pa{\phi}{x^\nu} = g^{\mu \nu} \nabla_\nu \phi $$
and
$$ \nabla_\mu p^\mu = - m^2 \phi $$
where the $\nabla$ is the Levi-Civita connection.
Putting these two together and remembering that the Levi-Civita connection is metric comptabible yields
$$g^{\mu \nu} \nabla_\mu \nabla_\nu \phi + m^2 \phi = 0$$
which is the curved space-time Klein-Gordon equation.

This is certainly the right result, but it is important to understand where it has come from. By isolating the Lagrangian from the entire volume form, we have caused the Lagrangian (and therefore the Hamiltonian also) to ``forget" about the (potentially) non-flat nature of the underlying space-time. The correction from a non-flat space-time then comes from the vertical-valued tensor used to give us our Hamiltonian equations of motion in eq.\ \eqref{alternativehamiltons}. If we had singled out the Lagrangian as everything expect the top form $d^4 x$ instead, we would have over-corrected for the underlying space-time. We will see that a similar procedure works equally well for electrodynamics, too.

\subsection{Electromagnetism} \label{maxwellfield}
In the case of electromagnetism, the most basic extended configuration space is $T^* \mathbb{R}^4 \simeq \mathbb{R}^4 \times \mathbb{R}^4$, so that $P$ has global fibered coordinates $\{ x^\alpha, A_\alpha, p^{\alpha \beta} \}$ and a point $p \in P $ is given by $p = x^\alpha e_\alpha + A_\alpha dx^\alpha + p^{\alpha \beta} dA_\alpha \otimes \pa{}{x^\beta} $. Given the Lagrangian density $\mathscr{L} = - \frac{1}{4}g^{\alpha \gamma} g^{\beta \delta} F_{\alpha \beta} F_{\gamma \delta} - g^{\alpha \beta} A_\alpha J_\beta$ (where $F_{\alpha \beta} := \pa{A_\beta}{x^\alpha} - \pa{A_\alpha}{x^\beta}$) care must be taken in how one proceeds in order to successfully carry out the covariant Legendre transformation and recover the correct equations of motion; see \cite{struckmeier2008covariant}. One finds that
$$ p^{\alpha \beta} := \pa{ \mathscr{L}}{ \pa{A_\alpha}{x^\beta}} = g^{\alpha \gamma} g^{\beta \delta} F_{\gamma \delta}$$
and so 
$$ H = -\frac{1}{4} g_{\alpha \gamma} g_{\beta \delta} p^{\alpha \beta} p^{\gamma \delta} + g^{\alpha \beta} A_\alpha J_\beta $$
Assuming a flat vertical projector, eq.\ \eqref{localhfe1} gives:
$$ \pa{A_\alpha}{x^\beta} = \pa{H}{p^{\alpha \beta}} = - \frac{1}{2} g_{\alpha \gamma} g_{\beta \delta} p^{\gamma \delta} = - \frac{1}{2} F_{\alpha \beta} = + \frac{1}{2} F_{\beta \alpha} $$
so that we have
$$ F_{\alpha \beta} = \pa{A_\beta}{x^\alpha} - \pa{A_\alpha}{x^\beta}$$
while eq.\ \eqref{localhfe2} gives:
$$ \pa{p^{\alpha \beta}}{x^\alpha} = \pa{F^{\alpha \beta}}{x^\alpha} = g^{\alpha \beta}J_{\beta}$$
which is the usual inhomogeneous Maxwell equation. \\

However, the alternative approach of section \ref{modifications} can be applied here, too, in order to extend this result to the non-flat case. Again, the presence of the background metric $g_{\alpha \beta}$ in the definition of the Maxwell Lagrangian density gives us access to the unique Levi-Civita connection $\nabla$ compatible with $g$; we shall use this connection to define the appropriate vertically-valued tensor in eq.\ \eqref{alternativehamiltons}. 

As before, we must isolate $\mathscr{L}$ from the entire volume form, writing
$$ S = \int \text{vol} \ \mathscr{L} = \int \sqrt{- \det g} \ d^4 x \ \left( - \frac{1}{4} g^{\mu \rho} g^{\nu \sigma} F_{\mu \nu} F_{\rho \sigma} - g^{\mu \nu} A_\mu J_\nu \right) $$
where as before we have
$$ F_{\mu \nu} := \pa{A_\nu}{x^\mu} - \pa{A_\mu}{x^\nu} $$
Then we have
$$ p^{\mu \nu} := \pa{ \mathscr{L}}{ \pa{A_\mu}{x^\nu}} = g^{\mu \sigma} g^{\nu \rho} F_{\mu \nu} $$
so that
\begin{multline*} H := p^{\mu \nu} \pa{A_\nu}{x^\mu} - \mathscr{L} =  \frac{1}{2} p^{\mu \nu} (\pa{A_\nu}{x^\mu} - \pa{A_\mu}{x^\nu}) - \frac{1}{4} g^{\mu \rho} g^{\nu \sigma} F_{\mu \nu} F_{\rho \sigma} + g^{\mu \nu} A_\mu J_\nu \\
 = - \frac{1}{4} p^{\mu \nu} F_{\mu \nu} + g^{\mu \nu} A_\mu J_\nu = - \frac{1}{4} g_{\mu \sigma} g_{\nu \rho} p^{\mu \nu} p^{\sigma \rho} + g^{\mu \nu} A_\mu J_\nu
\end{multline*}
Given local fibered coordinates in which
$$ p \in P = x^\mu e_\mu + A_\mu dx^\mu + p^{\mu \nu} dA_\mu \otimes \pa{}{x^\nu} $$
we have
$$ \omega = dA_\mu \wedge dp^{\mu \nu} \otimes \pa{}{x^\nu} $$
To make use of eq.\ \eqref{alternativehamiltons}, we shall anti-symmetrize the Levi-Civita connection $\nabla$ on $E$ (making it indistinguishable from the exterior derivative except for a factor of $\frac{1}{2}$, since the Levi-Civita connection is torsion free), vertically lift the Levi-Civita connection to $VE$ as in proposition \ref{liftedconnection} and use its natural dual connection on $V^*E$, and we shall use the Levi-Civita connection itself $\nabla$ on $TM$ to define the connection $\nabla_P$ which in turn defines $\nabla_{VP}$. In coordinates in which $\gamma : M \to P \mid x^\mu e_\mu \mapsto A_\mu (x^\mu) dx^\mu + p^{\mu \nu} (x^\mu) d A_\mu \otimes \pa{}{x^\mu}$, we have
$$\nabla_P \gamma = \frac{1}{2} \left(\pa{A_\mu}{x^\nu} - \pa{A_\nu}{x^\mu} \right) dx^\nu \otimes dx^\mu + \left( \pa{p^{\mu \nu}}{x^\xi} + p^{\rho \nu} \Gamma^\mu_{\rho \xi} + p^{\mu \rho} \Gamma^\nu_{\rho \xi} \right) dx^\xi \otimes dA_\mu \otimes \pa{}{x^\nu}$$
and
$$ \text{vl}_\gamma(\gamma) = x^\mu e_\mu + A_\mu dx^\mu + p^{\mu \nu} d A_\mu \otimes \pa{}{x^\mu} + A_\mu \pa{}{A_\mu} + p^{\mu \nu} \pa{}{p^{\mu \nu}} $$
so that
$$ \nabla_{VP} ( \text{vl}_\gamma(\gamma) ) = \frac{1}{2} \left( \pa{A_\mu}{x^\nu} - \pa{A_\nu}{x^\mu} \right) dx^\nu \otimes \pa{}{A_\mu} + \left( \pa{p^{\mu \nu}}{x^\xi} + p^{\rho \nu} \Gamma^\mu_{\rho \xi} + p^{\mu \rho} \Gamma^\nu_{\rho \xi} \right) dx^\xi \otimes \pa{}{p^{\mu \nu}} $$
Eq.\ \eqref{alternativehamiltons} then gives
$$ \frac{1}{2} \left( \pa{A_\mu}{x^\nu} - \pa{A_\nu}{x^\mu} \right) dp^{\mu \nu} - \left( \pa{p^{\mu \nu}}{x^\nu} + p^{\xi \nu} \Gamma^\mu_{\xi \nu} + p^{\mu \xi} \Gamma^\nu_{\xi \nu} \right) dA_\mu = - \frac{1}{2} g_{\mu \sigma} g_{\nu \rho} p^{\mu \nu} dp^{\sigma \rho} + g^{\mu \nu} J_\nu dA_\mu $$
which is equivalent to
$$ g^{\mu \sigma} g^{\nu \rho} \left( \pa{A_\rho}{x^\sigma} - \pa{A_\sigma}{x^\rho} \right) = p^{\mu \nu} = F^{\mu \nu} $$
and
$$ - \left( \pa{p^{\mu \nu}}{x^\nu} + p^{\xi \nu} \Gamma^\mu_{\xi \nu} + p^{\mu \xi} \Gamma^\nu_{\xi \nu} \right) = \nabla_\nu F^{\nu \mu} = g^{\mu \nu} J^\nu  $$
The first is the usual definition of the Maxwell tensor, while the second is the correct inhomogeneous Maxwell equation in curved space-time. 

It is worth pointing out explicitly that this construction is slightly less natural than the one in \ref{scalarfield}: we only recover the correct equations of motion if we anti-symmetrize the Levi-Civita connection to produce the connection $\nabla_E$ used in the construction of $\nabla_P$; there were no such complications in the case of the scalar field. This seems to indicate that there is something of an art to this approach, something which would be nice to avoid. Exactly how much additional challenge this represents in other cases is an interesting question for future research.

\subsection{General Relativity?}
Every covariant Hamiltonian formulation seems to struggle with the case of general relativity \cite{gotay1998momentum}. Fundamentally, this comes from the fact that the Legendre transformation of the Einstein-Hilbert Lagrangian density (considered as a function of the metric $g$) does not access all the information encoded in that function, as the Lagrangian density contains second derivatives of the metric which the Hamiltonian theory never becomes ``aware" of. Many attempts have been made to solve this problem, for instance by introducing higher order analogs of the Legendre transformation \cite{magnano1990legendre} or by using a different mechanism to produce a covariant Hamiltonian \cite{rovelli2006note} \cite{mcclain2018some}. To my knowledge, none of these attempts has been completely successful. Though it is disappointing, it should therefore be no great surprise that the Legendre transformation fails in this formalism, too. \\

Beginning with the standard Einstein-Hilbert action but identifying $\mathscr{L}$ as before gives
$$ S_{EH} = \int \text{vol} \ \mathscr{L} = \frac{c^4}{16 \pi G} \int \frac{d^4 x}{\sqrt{- \det g}} \ R $$
To evaluate the covariant momenta, we need to recall that, in order of increasing complexity
$$R = g^{\alpha \beta} R_{\alpha \beta}$$
$$ R_{\alpha \beta} = R^{\gamma}_{\alpha \gamma \beta} $$
$$ R^\alpha_{\beta \gamma \delta} = \pa{\Gamma^\alpha_{\beta \delta}}{x^\gamma} - \pa{\Gamma^\alpha_{\beta \gamma}}{x^\delta} + \Gamma^\alpha_{\epsilon \gamma} \Gamma^\epsilon_{\beta \delta} - \Gamma^\epsilon_{\beta \gamma} \Gamma^{\alpha}_{\epsilon \delta}$$
$$ \Gamma^\alpha_{\beta \gamma} = g^{\alpha \epsilon} \Gamma_{\epsilon \beta \gamma} = \frac{1}{2} g^{\alpha \epsilon} \left( - \pa{g_{\beta \gamma}}{x^\epsilon} + \pa{g_{\epsilon \beta}}{x^\gamma} + \pa{g_{\gamma \epsilon}}{x^\beta} \right)$$
Remembering that this is not a variational problem, we will need the identity
$$ \pa{g^{\alpha \beta}}{x^\gamma} = - g^{\alpha \delta} g^{\beta \epsilon} \pa{g_{\delta \epsilon}}{x^\gamma} $$
Together will a great deal of index gymnastics, this chain of equalities gives
\begin{multline}
\pi^{\zeta \eta \theta} := \pa{\mathscr{L}}{\pa{g_{\zeta \eta}}{x^\theta}} = \frac{c^4}{16 \pi G} \pa{R}{\pa{g_{\zeta \eta}}{x^\theta}} = \frac{c^4}{16 \pi G } g^{\beta \delta} \delta^\gamma_\alpha \pa{R^\alpha_{\beta \gamma \delta}}{\pa{g_{\zeta \eta}}{x^\theta}} \\
= \frac{c^4}{32 \pi G} \pa{g_{\alpha \beta}}{x^\gamma} \times \bigg(
- 4 g^{\alpha \eta } g^{\beta \gamma} g^{\zeta \theta} 
+ 2 g^{\alpha \zeta } g^{\beta \eta } g^{\gamma \theta}
- 2 g^{\alpha \beta } g^{\gamma \theta} g^{\zeta \eta }
+ g^{\alpha \beta } g^{\gamma \eta } g^{\zeta \theta} 
+ g^{\alpha \gamma} g^{\beta \theta} g^{\zeta \eta } 
 \bigg)
\end{multline}
Unfortunately, this is not invertible to get $\pa{g_{\alpha \beta}}{x^\gamma}$ in terms of $\pi^{\zeta \eta \theta}$. Even if it were, it would not give us access to the second derivatives $\pa{^2 g_{\alpha \beta}}{x^\gamma \partial x^\delta}$ that appear in the $\pa{\Gamma^\alpha_{\beta \delta}}{x^\gamma}$ terms in the Einstein-Hilbert action. The Legendre transformation fails. \\

In light of the aspirational perspective I mentioned in the introduction, this is a serious problem: one of the primary reasons to try to formulate a geometric approach to the quantization of classical fields is to feed in general relativity with the hope that the results will be better defined than those of the canonical approach. However, for this to be reasonable it is at minimum necessary that the classical theory be well defined within the framework used for quantization! The special case of general relativity is therefore an important area for future work within the field of polysymplectic Hamiltonian field theory.

\section{Conclusions}

In this paper, I have presented an extension of the polysymplectic formulation of covariant Hamiltonian field theory based on vertical projection operators that reproduces many of the key results of \cite{struckmeier2008covariant} -- most importantly Hamilton's field equations, the form of the generalized Poisson structure, and the covariant Hamiltonian -- in a more-or-less natural manner. In addition, the path I have taken generalizes G\"unther's approach in \cite{gunther1987polysymplectic} to the global case in a manner that is purely geometric and that seems to me much more in keeping with G\"unther's original approach and spirit than most recent work. This accounts for two of the three motivations I outlined in the introduction to justify this enterprise. 

The third motivation I offered for this particular treatment of covariant Hamiltonian field theory was that it might make possible a novel approach to the geometric quantization of classical field theories. Though I have not yet successfully used this framework to reproduce all the results of even scalar quantum field theory, there are a few promising points worth mentioning. First, it is a relatively simple matter to reproduce the canonical commutation relations for standard \emph{particle} theories using the geometric structures of this paper; indeed, the analysis can be successfully extended substantially farther, depending upon exactly which desirable properties one wishes the quantization procedure to reproduce \cite{mcclain2018quantization}. (That it will never produce all the properties one might reasonably expect for all smooth functions on the phase space is a result of Groenewald's theorem \cite{groenewold1946principles}.) On the field theory side, this procedure does not produce the canonical commutation relations. However, it does naturally reproduce the results of the canonical commutation relations \emph{after integration}; see \cite{mcclain2018some} for these preliminary results in the context of Minkowski space-time. This is a more reasonable expectation for a finite dimensional, differential geometric analysis of quantum field theory, as it is not clear how operator-valued distributions would ever arise from a well-defined differential geometric structure (see, for example, \cite{kanatchikov1998toward} and \cite{kanatchikov2001precanonical} for similar efforts, and \cite{blacker2019quantization} for a more closely related, though very mathematical, perspective). This seems a promising beginning, and a full analysis of the possibilities for geometric quantization from the perspective of this particular covariant Hamiltonian framework will be an important area for future work. It will be particularly interesting to see what role (if any) the vertical projection operator used to define Hamilton's field equations plays in defining the quantum theory in curved space-times.

\appendix

\section{Vertical projection operators}

Since vertical projection operators play such a major role in the construction of this version of polysymplectic covariant Hamiltonian field theory -- see, for example, definition \ref{tautologicaltensordefinition} and eq.\ \eqref{hamiltonsfieldequationsgeometric} -- it is worth briefly reviewing their properties for those unfamiliar with them.

Formally, a vertical projection operator on a fiber bundle $E$ over a base manifold $M$ is defined as a linear map $V : TE \to TE$ such that the following properties hold:

\begin{enumerate}
\item $ \text{Im} V = VE $
\item $ V^2 = V $
\end{enumerate}

More intuitively, the first property means both that the vertical projection operator must project onto the vertical space at each point and that it hit every element in each vertical space, while the second property means that the vertical projection operator is indeed a projection operator: once it has been applied once, that's all there is to it. \\

Since this paper deals exclusively with vector bundles, let us consider the specialized case of a vertical projection operator on a vector bundle more carefully. Let us choose coordinates $\{ x^\alpha, \phi^I \}$ for $E$ such that an arbitrary tangent vector $u \in T_e E$ can be represented as $u = u^\alpha \pa{}{x^\alpha} + u^I \pa{}{\phi^I}$. Then any endomorphism $E : TE \to TE$ can be written locally as
$$ E = E^\alpha_\beta(e) \pa{}{x^\alpha} \otimes dx^\beta + E^\alpha_I(e) \pa{}{x^\alpha} \otimes d\phi^I + E^I_\alpha(e) \pa{}{\phi^I} \otimes dx^\alpha + E^I_J(e) \pa{}{\phi^I} \otimes d \phi^J $$
If these local coordinates are adapted so that $v \in V_e E = v^I \pa{}{\phi^I}$, then the first criteria for vertical projection operators implies that we must have $E^\alpha_\beta = E^\alpha_I = 0$ if $E$ is to represent a vertical projection operator. Calling our prospective vertical projection operator $V$ instead of $E$, this means that we have
$$ V = V^I_\alpha(e) \pa{}{\phi^I} \otimes dx^\alpha + V^I_J(e) \pa{}{\phi^I} \otimes d \phi^J$$
The second criteria now implies that we must have
$$ V^2 = \pa{}{\phi^I} \otimes (V^I_J V^J_i dx^i + V^I_J V^J_K d \phi^K) $$
This is satisfied whenever $V^I_J = \delta^I_J$. In this case, we have
$$ V = \pa{}{\phi^I} \otimes d \phi^I + V^I_i \pa{}{\phi^I} \otimes dx^i $$
Specifying the arbitrary coefficients $V^I_i$ therefore specifies a vertical projection operator in adapted local coordinates. The case dealt with in the main article is similar, just a larger vertical space $VP$ in place of $VE$. In the local canonical coordinates used throughout most of the main article, this means that a vertical projection operator $V$ looks like
$$ V = \pa{}{\phi^I} \otimes (V^I_i dx^i + d \phi^I)  + \pa{}{\pi^i_I} \otimes (V^i_{I \, j} dx^j + d \pi^i_I)$$

\section{Vertical lifts on vector bundles} \label{verticallifts}

There is a standard construction for mapping two vectors $u, v$ in a fiber $E_x$ of a vector bundle $E$ to a vector $w$ in the vertical bundle $VE$. In coordinate independent terms, the desired vector $w$ is the tangent vector such that, for all functions $f : E \to \mathbb{R}$ we have
$$ w(f) := \frac{d}{ds} f(u + sv)$$
This tangent vector naturally lives in the space $V_e E \subset T_e E$, where the point $e \in E$ is given in terms of the vector $u$ by $e = (x, u) $. \\

Applying this construction to a single section $\gamma : M \to E$, with each $\gamma(x)$ serving the roles of both $u$ and $v$ gives
$$\text{vl} : \Gamma(M,E) \to \Gamma(M,VE) \mid \gamma(x) \mapsto \text{vl}_{\gamma(x)} \gamma(x)$$
In coordinates in which $\gamma(x) = x^\alpha e_\alpha + \phi^I(x) e_I $, this gives
$$ \text{vl}_\gamma \gamma = \phi^I(x) \pa{}{\phi^I}\bigg\rvert_{\gamma(x)} $$

\bibliography{sscft}

\end{document}